\documentclass[pra,aps,superscriptaddress,nofootinbib,twocolumn,longbibliography]{revtex4-2}

\usepackage[colorlinks=true,
    linkcolor=red,
    citecolor=blue,
    urlcolor=blue]{hyperref}
\usepackage[pdftex]{graphicx}
\usepackage{mathrsfs}
\usepackage{amsthm}
\usepackage{float}
\usepackage{comment}
\usepackage{cancel}
\usepackage{physics}
\usepackage[safe]{tipa}
\usepackage{textcomp}
\usepackage{fontenc}
\usepackage{mathtools}
\usepackage{mathrsfs}
\usepackage{color}
\usepackage{colortbl}
\usepackage{graphics,graphicx}
\usepackage{txfonts}
\usepackage{lipsum, babel}
\usepackage{bbm}
\usepackage{enumerate}
\usepackage{enumitem}
\usepackage{adjustbox}
\usepackage{footnote}
\usepackage[usenames,dvipsnames]{xcolor}
\usepackage{framed}
\usepackage{soul}

\newcommand{\cE}{\mathcal{E}}

\newcommand{\cH}{\mathcal{H}}

\newcommand{\cQ}{\mathcal{Q}}

\newcommand{\Id}{\mathbbm{1}}
\newcommand{\cg}[1]{#1_M^{\mathrm{cg}}}
\newcommand{\OEA}{observational ergotropic advantage }
\newcommand{\OBE}{observational ergotropy }
\newcommand{\IE}{incoherent ergotropy }
\newcommand{\orcid}[1]{\href{https://orcid.org/#1}{\includegraphics[width=10pt]{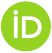}}}

\newtheorem{proposition}{Proposition}

\newtheorem{definition}{Definition}
\newtheorem{theorem}{Theorem}
\newtheorem{lemma}{Lemma}

\newtheorem{example}{Example}

\newtheorem{remark}{Remark}

\definecolor{cream}{rgb}{1.0, 0.99, 0.82}
\definecolor{celadon}{rgb}{0.67, 0.88, 0.69}
\definecolor{beaublue}{rgb}{0.74, 0.83, 0.9}
\definecolor{shadecolor}{rgb}{1.0, 0.99, 0.82}
\definecolor{agcol}{rgb}{0.9, 1.0, 0.9}

\begin{document}
\title{All coherent measurements provide observational ergotropic advantage}
\author{Soumik Mahanti\orcid{0000-0002-0380-0324}}
\email{soumik.mahanti@student.uts.edu.au}
\affiliation{School of Computer Science, University of Technology Sydney,
Ultimo, Sydney, New South Wales 2007, Australia}
\author{Rakesh Saini\orcid{0009-0005-1763-299X}}
\affiliation{School of Mathematical and Physical Sciences, Macquarie University, Sydney NSW 2113, Australia}
\author{Alexei Gilchrist\orcid{0000-0003-0075-5174}}
\affiliation{School of Mathematical and Physical Sciences, Macquarie University, Sydney NSW 2113, Australia}
\author{Arindam Mitra\orcid{0000-0001-8393-8886}}
\email{am56@iitbbs.ac.in}
\email{arindammitra143@gmail.com}
\affiliation{Department of Physics, School of Basic Sciences, Indian Institute of Technology Bhubaneswar, Odisha 752050, India.}

\begin{abstract}
The use of quantum resources in the work extraction from a quantum system is an emerging research topic. Recently, [arXiv:2602.22893] established the necessity of measurement coherence for obtaining advantage in work extraction from an unknown isolated quantum system. However, a quantitative relation between the magnitude of this advantage and established measures of measurement coherence is missing. Here, we establish such a connection by introducing a faithful operational quantifier of the work advantage provided by a measurement. We show that measurement coherence is necessary and sufficient for a positive advantage, and derive upper and lower bounds in terms of the robustness of measurement coherence and an $l_{\infty}$ norm based coherence measure respectively.  Finally, we examine this quantifier of advantage from the resource theoretic perspective. Our results provide an operational characterisation of measurement coherence as a resource for work extraction and reveal a nontrivial relation between its resource content and thermodynamic value.

\end{abstract}
\maketitle

\section{Introduction}
In the quantum world, there are several elements that do not have a classical analog. Examples of such elements include entanglement~\cite{Horodecki_Entanglement}, coherence~\cite{Streltsov_coherence}, measurement incompatibility~\cite{Heinosaari_meas_incomp,Uola_meas_incomp,Buscemi_meas_incomp_QRT}. These quantum features provide an advantage in several information-theoretic or thermodynamic tasks and thereby elevate 
quantum technology \cite{PhysRevA.106.L050401,PhysRevA.85.052308,QUACH20232195,VERMA2026116574,PhysRevLett.112.030602,tan2017quantum,PhysRevLett.128.200501,PhysRevApplied.20.044077}. For this reason, such quantum features are also known as quantum resources~\cite{Gour_QRT}. The role of state-based and state-assemblage based resources in information-theoretic and thermodynamic tasks has been widely studied. For example, entanglement enables teleportation~\cite{Nielsen_Chuang_2010}, dense coding~\cite{Nielsen_Chuang_2010}, and also provides an advantage in state discrimination, work extraction~\cite{roy2026}, etc. Another example is quantum coherence which is a resource in some quantum algorithms~\cite{PhysRevA.93.012111}, and provides advantages in work extraction~\cite{Korzekwa_2016} and channel discrimination~\cite{PhysRevA.93.042107}. Similarly, steerability, which is a state-assemblage based resource~\cite{PhysRevX.5.041008} provides security in one-sided device-independent QKD~\cite{PhysRevA.85.010301} and an advantage in work extraction~\cite{9qcc-7lq5} task. Apart from these well-explored resources, there exist measurement-based~\cite{wwn7-j8df,Tendick2023distancebased,Buscemi_meas_incomp_QRT,mitra2022quantifying,Baek_2020_Quantifying,baek2026maximalcoherencequantummeasurement,ao2026witnessrobustnessoperationalquantifier}, channel-based~\cite{PhysRevA.101.032331,PhysRevLett.125.180505,PhysRevA.100.052311,PhysRevA.109.062213}, and instrument-based quantum resources~\cite{ghai2026instrumentbasedquantumresourcesquantification,PRXQuantum.5.010340,mitra2026genuinecertificationincompatiblequantum} that have been comparatively less explored. Prominent examples of measurement-based resources include measurement incompatibility, measurement coherence, and measurement sharpness. While measurement incompatibility has received considerable attention, the operational significance of other measurement based resources is less fully understood. Measurement coherence is one such resource, and the application of measurement coherence in thermodynamic tasks is a comparatively open research direction, which is under focus here. 

The extraction of work from a quantum system is a very important aspect of modern-day quantum technologies that includes nanoscale energy storage \cite{QUACH20232195,VERMA2026116574}, nanoscale heat engine \cite{PhysRevLett.112.030602}, nanoscale refrigerator \cite{tan2017quantum}, enhancement of pricision measurement \cite{PhysRevLett.128.200501}, active initialization experiment of a superconducting qubit \cite{PhysRevApplied.20.044077}, etc. 
The study of work extraction from a microscopic system not only helps in the miniaturization of quantum devices but also helps to develop the understanding on the thermodynamic limit of the extractable work in quantum regime. The optimal work extractable from a closed quantum system is known as \emph{ergotropy}~\cite{Allahverdyan_2004}. To extract the optimal work from a system, one needs complete knowledge of the quantum state of the system. However, in practice, complete knowledge about the system requires full state tomography, which has a very high experimental cost, especially as the system dimension increases \cite{PhysRevLett.126.100402,PhysRevA.101.032321}. 
In quantum many-body systems, which are primary candidates for quantum batteries, tomographic requirements quickly become practically infeasible. There can be other related issues, including experimental limitation on performing certain kinds of measurements, measurement duration or the resource cost to perform it~\cite{Busch1990,Bernien2017,Strasberg_2022,Guryanova2020idealprojective}. To address these limitations, one can instead perform a suitable single measurement on the system to obtain partial knowledge about it and utilize it to extract maximum possible work. This leads to the notion of \emph{observational ergotropy}~\cite{PhysRevLett.130.210401}, which is a more useful quantity for practical purposes. However, the amount of extractable work then depends on the properties of the measurement as well as the input state. Hence, the quantum resources contained in quantum measurements play a vital role in the amount of maximally extractable work in this setting.


Recently, Ref.~\cite{Biswas_2026_Information} showed that coherence in quantum measurements may provide an advantage in observational ergotropy. However, that work did not quantitatively relate this advantage to any measure of measurement coherence. Motivated by this observation, we study the quantitative relationship between observational ergotropic advantage and measurement coherence. We introduce a state-independent quantifier of the advantage and derive upper and lower bounds in terms of measures of measurement coherence. We observe that measurement coherence is not only necessary, but also sufficient for a positive observational ergotropic advantage. Furthermore, we prove the tightness of these bounds for the qubit case and characterise this quantifier from a resource-theoretic perspective.

The rest of the paper is organized as follows. In Sec.~\ref{Sec:Prelims}, we discuss all relevant preliminaries. In Sec.~\ref{Sec:main}, we present our main results. In Sec.~\ref{Sec:conc}, we summarize our findings and discuss future directions.

\section{Preliminaries}
\label{Sec:Prelims}
\subsection{Measurement and measurement coherence}
A quantum measurement is a resolution of the identity into positive semi-definite operators, known as effects~\cite{Heinosaari_Ziman_2011}. More precisely, a set of operators $M:=\{M(\mu) \in \mathcal{L(H)}\}_{\mu \in \Omega_M}$ constitutes a measurement (also known as POVM) if $M(\mu) \geq 0, \text{ }\forall \mu \in \Omega_M$ and $\sum_\mu M(\mu) = \mathbbm{1}_\mathcal{H}$. Here, $\Omega_M$ is the set of outcomes for measurement $M$, and $\mathcal{L(H)}$ is the space of linear operators on the Hilbert space $\mathcal{H}$. If a measurement $M$ is performed on a state $\rho$, the probability of observing a particular outcome $\mu$ is given by the Born rule: $p(\mu|\rho, M) = \Tr(\rho M(\mu))$. In this work, we assume the dimension of $\mathcal{H}$ to be finite and the cardinality of the outcome set to be $n$. The set  of $n$-outcome measurements on a $d$ dimensional Hilbert space $\cH$ is denoted by $\mathscr{M}(\cH_d, n)$. An important class of measurements is the class of rank-1 projective measurements, where all the effects of a measurement are rank-1 projectors. Mathematically, $M=\{P_i\}_{i=0}^{d-1} : \sum_{i=0}^{d-1} P_i = \Id$ is a rank-1 projective measurement if $P_iP_j = \delta_{ij} P_i~\forall i,j$ and rank$(P_i) = 1~\forall i$. The set of rank-1 projective measurements is denoted as $\mathscr{P}(\cH_d,d)$. A rank-1 projective measurement is also known as a fine-grained measurement~\cite{Biswas_2026_Information}. Classical post-processing or coarse-graining of the fine-grained measurement is defined as 
\begin{equation}\label{Eq::Coarse-grained Measurement}
    Q_j = \sum_{i=0}^{d-1} D_{ji}P_i,~\forall j,
\end{equation}
where $\{Q_j\}_{j=0}^{n-1}$ is the \textit{coarse-grained} $n$-outcome POVM obtained from the fine-grained measurement $\{P_i\}_{j=0}^{n-1}$ via a column-stochastic matrix $D_{ji}:D_{ji}\geq 0~ \forall j,i~\text{and}~\sum_jD_{ji} =1$. This corresponds to assigning the final outcome $j$ by classical post-processing of the original outcome $i$, with conditional probability $D_{ji}$. This framework models a situation where an underlying rank-1 projective measurement is performed perfectly, but the resulting outcomes undergoes classical noise, blurring, or data reduction before being recorded. Next, we shift our focus to coherence in measurement and measures of measurement coherence.

Conventionally, state coherence refers to the presence of non-zero off-diagonal elements in a density matrix, indicating quantum superposition between different basis states. Similar to quantum states, coherence can be present in quantum measurements as well. Given an arbitrary basis of $\mathcal{H}$, one can represent any operator in $\mathcal{L(H)}$ as a $d\times d$ matrix. If all the effects of a measurement are diagonal in that basis, then the measurement is said to be incoherent. Mathematically, if $M(\mu) = \sum_i m_{\mu,i}\ketbra{i}{i},~\forall \mu$ in a given orthonormal basis $\{\ket{i}\}_{i=0}^{d-1}$, the measurement is said to be incoherent with respect to that basis. It is easy to see that any incoherent measurement can be obtained via classical post-processing of the rank-1 projective measurement with basis-projectors $\{\ketbra{i}{i}\}_{i=0}^{d-1}$. We denote the set of incoherent measurements with $n$-outcome as $\mathscr{I}(\cH_d, n)$, while the whole set of incoherent measurements with any number of outcomes is denoted as $\mathscr{I}(\cH_d)$.

To quantify the coherence associated with a quantum measurement, we employ the resource theory of measurement coherence introduced in Ref.~\cite{Baek_2020_Quantifying}. Unlike quantum states, quantum measurements can affect coherence in different ways. They may destroy the coherence present in a coherent state or, conversely, generate coherence when acting on an incoherent state. A more detailed discussion of the resource theory of measurement coherence, together with its specific role in our analysis, is provided in Appendix~\ref{Counter example of the PIO}. We enlist a few  measures of measurement coherence that faithfully quantify the amount of coherence present in a given measurement in the following. 
\begin{itemize}
    \item \emph{Robustness of measurement coherence:} For any measurement $M$, it is defined as~\cite{Baek_2020_Quantifying} 
    \begin{equation}\label{Eq::Robustness of coherence}
        R_C(M) := \min_{r\geq 0}\left\{\frac{M+rP}{1+r} =N\in \mathscr{I}(\cH_d, n)\}\right\}.
    \end{equation}
    
    For a given measurement $M \in \mathscr{M}(\cH_d, n)$, it captures the amount of minimal mixing with any measurement $P \in \mathscr{M}(\cH_d, n)$ that makes the resultant measurement $N$ an incoherent one. The robustness can be understood in the form of a convex optimisation problem where any $r$ that satisfies the relation in Eq.~\eqref{Eq::Robustness of coherence} is known as a feasible solution and the set of all feasible solutions is known as the feasible set. The robustness is a measure that does not depend on the choice of free operations once the free set is fixed.
    
    \item \emph{$l_\infty$-norm based measure:} For any measurement $M$, it is defined as~\cite{Baek_2020_Quantifying} 
    \begin{equation}\label{Eq::Infinity-norm based coherence}
        C_{l_\infty}(M) := \max_{i<j}\sum_{\mu=0}^{n-1}\left|\bra{i}M(\mu)\ket{j}\right|.
    \end{equation}
For a given measurement $M$ and basis $\{\ket{i}\}_{i=0}^{d-1}$, the quantity $C_{l_\infty}(M)$ captures the largest total magnitude of the off-diagonal elements of the measurement effects between any pair of basis states, summed over all measurement outcomes. Unlike the robustness, $C_{l_\infty}(M)$ can be evaluated directly from the matrix elements of the measurement effects without requiring an additional optimisation.
\end{itemize}
\subsection{Quantum Channel}
In the Schr\"odinger picture, a quantum channel is a map that transforms any quantum state into a valid quantum state. Mathematically, a linear map $\Lambda: \mathcal{L(H)} \rightarrow \mathcal{L(K)} $ that is completely positive and trace-preserving (CPTP) represents a quantum channel. Whereas in the Heisenberg picture, the state is unchanged but the operators evolve through the corresponding dual map $\Lambda^\dagger : \mathcal{L(K)} \rightarrow \mathcal{L(H)}$. The action of the dual map $\Lambda^\dagger$ is defined as $\Tr(A^\dagger\Lambda(B))=\Tr((\Lambda^\dagger(A))^\dagger B)$; or, equivalently, $\Tr(A(\Lambda(B))^\dagger)=\Tr( B^\dagger \Lambda^\dagger(A))$. Here, $B \in \mathcal{L(H)}$, and $A \in \mathcal{L(K)}$ are arbitrary.  Note that the probabilities are still given by the Born rule in both pictures $p(\mu|\Lambda(\rho)) = \Tr(\Lambda(\rho) M(\mu)) = \Tr(\Lambda^\dagger(M(\mu))\rho )$. 
In the Heisenberg picture, the action of the dual map on the measurement is denoted as the following: $\Lambda^\dagger(M) := \{\Lambda^\dagger(M(\mu))\}_{\mu \in \Omega_M}$. 
\subsection{Ergotropy and work extraction}
For any quantum state $\rho$, the maximum amount of extractable work in a closed-system setting is called ergotropy. Mathematically, ergotropy is defined as 
\begin{equation}\label{Eq::Ergotropy formal}
   \mathcal{E}(\rho) := \Tr(H\rho) - \min_U \Tr(HU\rho U^\dagger), 
\end{equation}
 where $H$ is the associated Hamiltonian of the system. The state that has the minimum energy among all the achievable states under unitary transformation of the input state is known as the \textit{passive state} of the input. This also serves as the definition of the passive state $\Pi_\rho$ as $\Tr(H\Pi_\rho) = \min_U \Tr(HU\rho U^\dagger)$. The function that gives the energy of the passive state, called the passive energy, is denoted as $P_H$. Formally, it is defined as $P_H(\rho) := \min_U \Tr(HU\rho U^\dagger)$. For a $d$-dimensional Hamiltonian $H=\sum_{i=0}^{d-1} E_i\ketbra{E_i}{E_i}$, where $\{\ket{E_i}_{i=0}^{d-1}\}$ forms the energy eigenbasis with eigenvalues in increasing order:  $E_0\leq E_1\leq\cdots\leq E_{d-1}$, the passive energy of a state is equal to 
 \begin{equation}\label{Eq::passive energy}
   P_H(\rho) = \sum_{i=0}^{d-1}E_i\lambda_i^\downarrow(\rho),  
 \end{equation}
 where $\lambda_i^\downarrow(\rho)$ are the eigenvalues of the state arranged in descending order, i.e., $\lambda_0^\downarrow\geq\ldots\geq \lambda_{d-1}^\downarrow$. 
 Hence, if one writes the density matrix of a state $\rho$ in the energy eigenbasis, the ergotropy of the state equals (from Eq.'s~\eqref{Eq::Ergotropy formal} and \eqref{Eq::passive energy}) 
 \begin{align}\label{Eq::Ergotropy calculational}
    \mathcal{E}(\rho) &= \Tr(H\rho)-P_H(\rho)\nonumber\\
    \mathcal{E}(\rho) &= \sum_i E_i(\rho_{ii}-\lambda_i^\downarrow(\rho)), 
 \end{align}
 with $\rho_{ii}$ as the diagonal entries in the chosen basis. This clearly shows that if a state is diagonal, then non-zero work extraction is possible only if the diagonal entries are not in a decreasing order. From here, one can make a distinction between two different contributions to ergotropy. One of them is completely classical in origin, called the incoherent ergotropy~\cite{Francica_2020_Quantum}, which gives us the maximum extractable work just by permuting the diagonal entries 
\begin{align}\label{Eq::Incoherent ergotropy}
  \mathcal{E}_{\rm{inc}}(\rho) &= \Tr(H\Delta\rho) - \min_U \Tr(HU\Delta\rho U^\dagger)\nonumber\\  
  &= \Tr(H\Delta\rho) - P_H(\Delta\rho).
\end{align}
We use the short notation $\Delta\rho$ to denote the completely dephased(diagonal) state under the action of completely dephasing map $\Delta$, defined as \[\Delta(\rho)=\Delta\rho = \sum_i\bra{E_i}\rho\ket{E_i}\ketbra{E_i}{E_i}.\]The incoherent ergotropy physically quantifies the maximum extractable work from the completely dephased state. While the other contribution to ergotropy, called the coherent ergotropy $\mathcal{E}_{\rm{co}}(\rho)=\mathcal{E}(\rho)-\mathcal{E}_{\rm{inc}}(\rho)$ is the extractable work coming from the coherence of a quantum state, which is a kind of resource that is completely quantum in origin. Ergotropy is interpreted as the maximum extractable work from a closed-system because the unitary operation is implemented by a cyclic driving Hamiltonian $H_d(t)$ as $U=\mathcal{T}\exp(\int_0^\tau-iH_d(t)dt/\hbar)$, where $\mathcal{T}$ is the time ordering operator, and $H_d(0) = H_d(\tau) = H$~\cite{Allahverdyan_2004}. 
The cyclic Hamiltonian drives the system back to the initial Hamiltonian and the energy difference between the final and the initial states is realized as the \emph{amount} of work extracted from the system. However, constructing the unitary that converts the input state into its passive state requires the complete knowledge of the state. To get complete information about an unknown input state, one typically needs quantum state tomography~\cite{MAURODARIANO2003205}, which itself is a very resource expensive task. An alternative option is to choose a single measurement (need not be informationally complete)
to probe the state and get partial information. This motivates the introduction of a more operational quantity named `observational ergotropy' that gives us the amount of accessible work when partial information about the input state is known through a single measurement~\cite{Safranek_2023_Work}.
Mathematically, \OBE is defined as the following: 
\begin{align}\label{Eq::Observational ergotropy}
  \mathcal{E_O}(\rho,M) &:= \Tr(H\rho) - \min_U \Tr(HU\cg{\rho} U^\dagger)\nonumber\\
  &= \Tr(H\rho) - P_H(\cg{\rho}),
\end{align}
 where $\cg{\rho}$ is called the coarse-grained state of the measurement $M$; defined as 
 \begin{equation}\label{Eq::coarse-grained state}
     \cg{\rho} := \sum_\mu \Tr(M(\mu) \rho)\frac{M(\mu)}{\Tr(M(\mu))}. 
 \end{equation} 
 The \OBE is upper bounded by the ergotropy for all measurements and input states~\cite{Biswas_2026_Information}
 \begin{equation}
     \mathcal{E_O}(\rho,M) \leq \mathcal{E}(\rho)~\forall \rho,M,
 \end{equation}
 and they are equal, if the coarse-grained state is same as the input state (see Eq.'s~\eqref{Eq::Ergotropy calculational} and \eqref{Eq::Observational ergotropy}).
 The coarse-grained state can also be expressed as 
 \begin{equation}
 \cg{\rho}=\sum_\mu p(\mu)W(\mu),
 \end{equation}
 where $p(\mu)=\Tr(\rho M(\mu))$ are the outcome probabilities and ${M(\mu)}/{\Tr(M(\mu))}=W(\mu)$ are the normalized effects. The associated channel with coarse-graining is denoted as $\Lambda_{M}^{cg}$, i.e, 
 \begin{equation}\label{Eq::Coarse-graining channel}
     \Lambda_{M}^{cg}(\rho) = \cg{\rho}.
 \end{equation}
 Channels that maps identity operators into itself are known as unital channels, and one can easily verify that $\cg{\Lambda}$ is unital. 

\section{Main Results}
\label{Sec:main}
The observational ergotropy depends on both the input state and the measurement used to probe it. Our aim is to characterize how useful a given measurement is for extracting work beyond the incoherent contribution and, in particular, to relate this operational advantage to the amount of coherence contained in the measurement. Since the observational ergotropy itself is state dependent, this requires a state-independent quantity that can be associated with the measurement alone. A natural benchmark is the maximum observational ergotropy attainable using incoherent measurements. It was shown in Ref.~\cite{Biswas_2026_Information} that, for any input state $\rho$, the \OBE for incoherent measurements are upper bounded by the incoherent ergotropy, i.e., 
\begin{equation*}
    \mathcal{E}_O(\rho,N)\leq \mathcal{E}_{\rm inc}(\rho),
    \qquad
    \forall N\in\mathscr{I}(\mathcal{H}_d).
\end{equation*}
Therefore, coherence in the measurement is necessary to obtain observational ergotropy beyond this incoherent limit. However, this result by itself does not establish whether every coherent measurement can provide such an advantage, nor does it quantify how the magnitude of the advantage depends on measurement coherence. These are the questions we address in the following. Analogous to ratio-based advantages commonly used in resource theories~\cite{Mori_2020_Operational,Uola_2020_All,biswas2025steerablequantumcorrelationsprovide,Skrzypczyk_2019_All}, one may consider the following observational ergotropic ratio to capture this advantage quantitatively
\begin{equation}\label{Eq::ergotropic quotient}
    \cQ_M^H
    =
    \max_{\rho}
    \frac{\mathcal{E_O}(\rho,M)}
    {\max\limits_{N\in \mathscr{I}(\cH_d)}\mathcal{E_O}(\rho,N)}
\end{equation}
We call this ratio the \textit{ergotropic quotient} of measurement $M$, where the numerator is \OBE for the measurement $M$ and the denominator is the maximum achievable \OBE for an incoherent measurement. The ratio is further maximised for all input states, rendering $\cQ_M^H$ as solely a property of the measurement $M$. A value $\cQ_M^H>1$ would indicate that the measurement enables
observational ergotropy beyond the contribution from any incoherent measurement for some input state. Therefore, $\cQ_M^H>1$ denotes obtainable quantum advantage in thermodynamic work extraction by the action of the measurement $M$. However,
in the following we establish why $\cQ_M^H$ is not a suitable quantifier of the
advantage. First, we note  a result from~\cite{Biswas_2026_Information} that the maximum \OBE for any incoherent measurement is the incoherent ergotropy.
\begin{equation}\label{Eq::maximum observ ergo for inc meas = inc ergo}
   \max_{N\in \mathscr{I}(\cH_d)}\mathcal{E_O}(\rho,N) = \mathcal{E}_{\rm{inc}}(\rho) 
\end{equation}
The problem arises with the quantity $\cQ_M^H$ (see Eq.~\eqref{Eq::ergotropic quotient}) as the \IE vanishes for a large class of states. From Eq.~\eqref{Eq::Incoherent ergotropy} and \eqref{Eq::passive energy}, it follows that the \IE for any state is 
\begin{equation}
    \cE_{\rm{inc}}(\rho) = \sum_{i=0}^{d-1}E_i(\rho_{ii}-\lambda_i^\downarrow(\Delta\rho))
\end{equation}
with $\rho_{ii}$ as the $i$-th diagonal element of the density matrix written in the energy basis. It is easy to see that if the dephased version of a state is passive ($\Pi_{\Delta\rho}=\Delta\rho$), or the diagonal entries of the state are arranged decreasingly, the \IE (denominator of Eq.~\eqref{Eq::ergotropic quotient}) goes to $0$, while the numerator can still be finite. So, $\cQ_M^H$ is undefined. Furthermore, even if one tries to modify the definition of the ergotropic quotient as
\begin{equation}
    \mathbbm{Q}_M^H = \sup_{\rho:\mathcal{E}_{\rm inc}(\rho)>0} \frac{\mathcal{E_O}(\rho,M)}{\mathcal{E}_{\rm inc}(\rho)},
\end{equation}
still the modified quotient $\mathbbm{Q}_M^H$ is not a good quantifier. We demonstrate this using an example of projective measurements for qubits under two-level Hamiltonian $H = E\ketbra{1}{1}$. A qubit state can be expressed as $\rho = \frac{1}{2}(I+\Vec{n}.\Vec{\sigma})$ with $|\vec{n}|\leq 1$, where $\Vec{\sigma}$ denotes the vector containing Pauli matrices. Similarly, any qubit rank-1 projective measurement can be written as $M=\{M(0),M(1)\}; M(0) = \frac{1}{2}(I+\Vec{m}.\Vec{\sigma}), M(1) = \frac{1}{2}(I-\Vec{m}.\Vec{\sigma})$ with $|\vec{m}|= 1$. One can explicitly calculate the observational ergotropy $\mathcal{E_O}(\rho,M)$ for this state and measurement (details in Appendix~\ref{App:observational ergotropic advantage for qubits}) which equals $\frac{E}{2}(|\Vec{n}.\Vec{m}|-n_z)$. While the incoherent ergotropy $\mathcal{E}_{\rm inc}(\rho)$ amounts to $\frac{E}{2}(|n_z|-n_z)$. When the denominator is nonzero ($n_z<0$), the modified quotient becomes
\begin{equation}\label{Eq::qubit modified ergotropic quotient}
    \mathbbm{Q}_M^H = \frac{1}{2}\left(1-\frac{|\Vec{n}.\Vec{m}|}{n_z}\right).
\end{equation} 
If the measurement is incoherent, then $\Vec{m} = (0,0,1)$; implying $\mathbbm{Q}_M^H$ goes to 1 (from Eq.~\eqref{Eq::qubit modified ergotropic quotient}). But, if the measurement has some coherence, then one can choose a state $\rho$ such that $n_z\rightarrow0^-$, while $|\Vec{n}.\Vec{m}|$ is still finite. Then $\mathbbm{Q}_M^H$ again diverges. So, for the qubit projective measurement, one finds that \[\sup_{\rho:\mathcal{E}_{inc}(\rho)>0} \frac{\mathcal{E_O}(\rho,M)}{\mathcal{E}_{inc}(\rho)}=
    \begin{cases}
        & 1, \text{ if $M$ is incoherent}\\
        & \infty, \text{ if  $M$ is coherent}
    \end{cases}
    \]
Therefore, the modified quotient is also not a good quantifier of the quantum advantage, in general. Hence, we make the following remark.
\begin{remark}\label{Ratio advantage does not work}
    The ratio of observational ergotropy and incoherent ergotropy, maximised over all states is not a faithful quantifier of the quantum advantage in ergotropic work extraction. 
\end{remark}
To \emph{avoid above-said technical problem}, we introduce an alternative quantity that faithfully captures this quantum advantage.
\begin{definition}\label{Def:Observational ergotropic advantage}
  {\textbf{Observational ergotropic advantage: }} For a measurement $M$, we define its observational ergotropic advantage as the largest excess in observational ergotropy that $M$ can achieve over all incoherent measurements, maximized over all input states. 
\begin{equation}\label{Eq::Observational ergotropic advantage general}
   A_M^H := \max_\rho \left[\mathcal{E_O}(\rho,M)-\max_{N\in \mathscr{I}(\cH_d)}\mathcal{E_O}(\rho,N)\right]. 
\end{equation}
\end{definition}
Using Eq.~\eqref{Eq::maximum observ ergo for inc meas = inc ergo}, the expression in Eq.~\eqref{Eq::Observational ergotropic advantage general} becomes 
\begin{equation}\label{Eq::Observational ergotropic advantage simple}
  A_M^H = \max_\rho \left[\mathcal{E_O}(\rho,M)- \mathcal{E}_{\rm{inc}}(\rho)\right].  
\end{equation}
The observational ergotropic advantage is non-negative for every measurement. Indeed, for the maximally mixed input \(\rho=\Id/d\), both the observational and incoherent ergotropies vanish, and hence $A_M^H \geq 0,~\forall M$. We note here that this newly introduced quantity \emph{observational ergotropic advantage} captures the quantum advantage in the same spirit as the quotients {\small$\cQ_M^H/\mathbbm{Q}_M^H$} would, if they were faithful. In the expression of $A_M^H$, we used difference of \OBE and \IE instead of ratio. Therefore, even when the \IE is 0, the difference would be unaffected. 
Note that, for ratio based quantifiers like {\small$\cQ_M^H/\mathbbm{Q}_M^H$}, when it is more than 1, the quantum advantage is certified; whereas a difference, when positive: $A_M^H>0$, certifies the same. Therefore, a positive \OEA($A_M^H$) certifies a quantum advantage.   
Combining Eq.~\eqref{Eq::Incoherent ergotropy} and \eqref{Eq::Observational ergotropy}, one can rewrite Eq.~\eqref{Eq::Observational ergotropic advantage simple} as 
\begin{equation}\label{Eq::Observational ergotropic advantage passive}
A_M^H = \max_\rho \left[P_H(\Delta\rho)- P_H(\cg{\rho})\right].      
\end{equation}
Therefore, operationally \OEA gives the maximum possible passive energy gap between the dephasing and coarse-graining of a state. Furthermore, we observe that the observational ergotropic advantage is non-increasing under coarse-graining of the measurement. If one starts with a rank-1 projective measurement $M \in \mathscr{P}(\cH_d,d)$, and under coarse-graining gets a POVM $M' \in \mathscr{M}(\cH_d,n)$, the \OEA is non-increasing as shown below. 
\begin{align*}
\mathcal{E_O}(\rho,M')&\leq \mathcal{E_O}(\rho,M),~ \forall \rho \qquad\left(\text{from \cite{Biswas_2026_Information}}\right)\\
\implies \mathcal{E_O}(\rho,M') - \mathcal{E}_{\rm{inc}}(\rho)&\leq \mathcal{E_O}(\rho,M)-\mathcal{E}_{\rm{inc}}(\rho),~ \forall \rho.
\end{align*}
Maximising both sides over all states, we get 
\begin{equation}\label{Eq::observational ergotropic advantage is monotonic under classical post-processing}
   A_{M'}^H\leq A_{M}^H. 
\end{equation} It implies that more informative measurements (fine-grained measurements) yield more quantum advantage. Next, we show that this quantity is faithful to the measurement coherence, meaning it is a non-negative quantity and is equal to 0 if and only if there is no coherence in the measurement.
\begin{theorem}\label{Th::Coherence necessary and sufficient}
    Coherence in measurement is both necessary and sufficient for a positive observational ergotropic advantage.
\end{theorem}
\begin{proof}
    Proving that coherence is necessary for a positive \OEA is easier, following a result in Ref.~\cite{Biswas_2026_Information}. If $M=\left\{ M(\mu)\right\}$ is any incoherent POVM, then $\mathcal{E_O}(\rho,M)\leq \mathcal{E}_{inc}(\rho)$ for an arbitrary state $\rho$. Therefore, $A_M^H\leq 0$. But for the maximally mixed state $\rho=\mathbbm{1}/d$, the coarse-grained state $\cg{\rho}$ and the dephased state $\Delta\rho$ both are maximally mixed. Therefore, $A_M^H=0$, which proves the necessary part. However, if at least one of the effects is not diagonal in the energy basis, that is $\exists \mu: M(\mu)-\Delta(M(\mu))\neq0$, we show that the \OBE is strictly positive. Take the input state $\rho=\mathbbm{1}/d+\epsilon X$, where $X = M(\mu)-\Delta(M(\mu))$. Note that it is already Hermitian and of unit trace and one can choose a suitable $\epsilon:~0<\epsilon<\frac{1}{d |\lambda_{\rm{min}}(X)|}$ to ensure positivity, where $\lambda_{\rm{min}}(X)$ is the minimum eigenvalue of $X$. Therefore, the input state is a valid state. The coarse-grained state and the dephased state are $\cg{\rho} = \mathbbm{1}/d+\epsilon \cg{X}$ and $\Delta\rho=\mathbbm{1}/d$ respectively, and the set of eigenvalues for the corresponding states are 
    \begin{align*}
       \{\lambda(\cg{\rho})\}&=\{1/d+\epsilon x_{\mu,i}\}~ \text{ and}\\
       \{\lambda(\Delta\rho)\}&=\{1/d\}, 
    \end{align*}
    where $\{ x_{\mu,i}\}$ are the eigenvalues of $\cg{X}$. Using Eq.~\eqref{Eq::passive energy}, one can find that the passive energy of the dephased state $P_H(\Delta\rho)=\Tr(H)/d$ and the passive energy of the coarse-grained state $P_H(\cg{\rho})={\Tr(H)}/{d} + \epsilon \sum_i E_ix^\downarrow_{\mu,i}$, with $x^\downarrow_{\mu,i}$ being the eigenvalues of $\cg{X}$ in decreasing order. We note here that $\cg{X} \neq 0$, but all the diagonal entries are 0. Since $\sum_i x^\downarrow_{\mu,i} = 0$ and the energies are arranged in increasing order,
    $\sum_i E_ix^\downarrow_{\mu,i} <0$. Therefore, $P_H(\Delta \rho)-P_H(\cg{\rho}) = -\epsilon\sum_i E_ix^\downarrow_{\mu,i} > 0$ implying  $A_M^H >0$. Hence, coherence is sufficient for the quantum advantage. 
\end{proof}
\noindent Therefore, we have shown that any coherent measurement is capable of providing advantage over incoherent ones in terms of work extraction. But the immediate next questions are: what is the maximum possible advantage and which measurements achieve it. We answer them in the following:  
\begin{proposition}
Let $H$ be a non-degenerate $d$-dimensional Hamiltonian with ground-state energy $E_0$. Then, for every measurement $M$,
\begin{equation}\label{Eq::OEA universal ceiling}
    A_M^H \leq \frac{\Tr(H)}{d} - E_0 .
\end{equation}
Equality holds if and only if there exists a subset of outcomes whose corresponding effects sum to a maximally coherent rank-one projector $\ketbra{\psi}{\psi}$, while all remaining effects have support orthogonal to $\ket{\psi}$. The maximum advantage is then attained by that maximally coherent input state $\ketbra{\psi}{\psi}$.

\end{proposition}
\begin{proof}
     The eigenvalue vector of any state $\rho$ majorises the uniform distribution $\lambda^\downarrow(\rho)\succ \lambda(\mathbbm{1}/d)$~\cite{Marshall_2011_Inequalities:}. Also, the passive energy function is Schur-concave (see End matter of Ref.~\cite{Biswas_2026_Information}). 
     Combining these two facts, it follows that $P_H(\rho)\leq P_H(\mathbbm{1}/d)$, with $P_H(\mathbbm{1}/d) = {\Tr(H)}/{d}=\Bar{E}$. Also, the passive energy of any state $\sigma$ cannot be lower than the ground state energy $P_H(\sigma)\geq E_0$. Therefore, for any two arbitrary states $\rho\text{ and }\sigma$, the passive energy difference can not be larger than this gap $P_H(\rho)-P_H(\sigma) \leq {\Tr(H)}/{d}-E_0$. Hence the \OEA is also upper bounded by this universal ceiling $A_M^H\leq {\Tr(H)}/{d}-E_0$. This proves the first part, with equality only when $P_H(\Delta \rho) = \Bar{E}$, and $P_H(\cg{\rho})=E_0$. Now, if the Hamiltonian is non-degenerate 
    : $~E_0<E_1<\cdots<E_{d-1}$; the passive energy of any state is equal to $\Bar{E}$ if and only if it is maximally mixed; while the passive energy is equal to $E_0$ only if the state is pure. Therefore, we obtain the maximum advantage for a measurement and an input state such that $\Delta\rho = \Id/d$, and $\cg{\rho} = \ketbra{\psi}{\psi}$ for some $\ket{\psi} \in \cH$. From Eq.~\ref{Eq::coarse-grained state}, this implies that $\sum_\mu p(\mu)W(\mu) = \ketbra{\psi}{\psi}$. Since, $W(\mu)$ are nothing but normalized states, this can only be true if 
     \begin{align}\label{Eq::maximum advantage}
         W(\mu) &=\ketbra{\psi}{\psi},~\forall \mu:p(\mu) >0\text{, and}\nonumber\\
         \sum_{\mu:p(\mu)>0} p(\mu) &= 1
     \end{align}
     Which means all the effects with non-zero probabilities must be proportional to the projector $\ketbra{\psi}{\psi}$, i.e.  $M(\mu)=\Tr(M(\mu))\ketbra{\psi}{\psi}$. From this and Eq.~\eqref{Eq::maximum advantage}, it follows
     \begin{equation}
         \sum_\mu \bra{\psi}\rho\ket{\psi}\Tr(M(\mu)) = 1
     \end{equation}
     which is only true, if
     \begin{align}
         \rho=\ketbra{\psi}{\psi}, \text{ and }& \sum_{\mu:p(\mu)>0} \Tr(M(\mu)) = 1\nonumber \\
         \implies \rho = \ketbra{\psi}{\psi},~\text{and }&~\sum_{\mu:\Tr(\rho M(\mu))>0}M(\mu) = \ketbra{\psi}{\psi} 
     \end{align}
    But since, $\Delta\rho = \Id/d$, we have $|\bra{E_i}\ket{\psi}|^2=1/d$, implying $\ket{\psi} = 1/\sqrt{d}\sum_i  e^{\phi_i}\ket{E_i}$, or the projector $\ketbra{\psi}{\psi}$ is maximally coherent in the energy eigenbasis. This proves that to get the maximum observational ergotropic advantage, some effects of the measurement must add up to a maximally coherent projector, while all the remaining effects have support orthogonal to it and the input state must be the same maximally coherent projector.  
\end{proof}
\begin{theorem}\label{Th::Coherence upper bound intermediate}
    For a measurement $M\in \mathscr{M}(\cH_d, n)$, the \OEA is upper bounded by the following quantity:
    \begin{equation*}
        A_M^H \leq \sum_{k=1}^{d-1} (E_k-E_{k-1}) \max_{\substack{
0 \leq Q \leq \mathbbm{1} \\
\Tr(Q)=k
}}
        \lambda_{\rm{max}}\left[\cg{\Lambda}(Q)-\Delta(\cg{\Lambda}(Q))\right]
    \end{equation*}
    where $\lambda_{\rm{max}}(A)$ denotes the maximum eigenvalue of an operator $A$.
\end{theorem}
\begin{proof}
 From Eq.~\eqref{Eq::Observational ergotropic advantage passive} and~Eq. \eqref{Eq::passive energy}, we get 
 \begin{equation}\label{Eq::observational ergotropic advantage eigenvalue}
     A_M^H = \max_\rho\left( \sum_{i=0}^{d-1} E_i(\lambda^\downarrow_i(\Delta\rho)-\lambda^\downarrow_i(\cg{\rho}))\right)
 \end{equation}
 where $\lambda_i^\downarrow(\sigma)$ denotes the eigenvalues arranged in decreasing order for a state $\sigma$. We note that 
 \begin{equation}
   \sum_{i=0}^{d-1}E_i\lambda^\downarrow_i(\sigma) = E_{d-1}-\sum_{k=1}^{d-1}(E_k-E_{k-1})S_{k-1}(\sigma),
 \end{equation}
where $S_{k-1}(\sigma)= \sum_{i=0}^{k-1} \lambda^\downarrow_i(\sigma)$, i.e., sum of $k$-largest eigenvalues of the state $\sigma$. Note that, $k$ takes all integer values except 0, i.e., $k=1,\ldots,d-1 $. Substituting this on Eq.~\eqref{Eq::observational ergotropic advantage eigenvalue}, one derive the following 
 \begin{align}\label{Eq::no_name_yet1}
     A_M^H &= \max_\rho\left( \sum_{k=1}^{d-1} (E_k-E_{k-1})\left[S_{k-1}(\cg{\rho})-S_{k-1}(\Delta\rho)\right]\right)\nonumber\\
     &\leq \sum_{k=1}^{d-1} (E_k-E_{k-1})\max_\rho\left[S_{k-1}(\cg{\Lambda}(\rho))-S_{k-1}(\Delta\rho)\right],
 \end{align}
where the inequality comes from maximising individual terms of difference between $k$-largest eigenvalues, and we have used $\cg{\rho} = \cg{\Lambda}(\rho)$. Now, one can find that the sum of $k$-largest eigenvalues of a state as 
\begin{equation}\label{Eq::k-largest eigenvalue}
 S_{k-1}(\sigma)= \max_{\substack{
0 \leq Q \leq \mathbbm{1} \\
\Tr(Q)=k
}}   \Tr(Q\sigma).
\end{equation}
Similarly, for a diagonal state $\Delta\sigma$, the sum of $k$-largest eigenvalues are
\begin{equation}\label{Eq::k-largest eigenvalue diagonal}
 S_{k-1}(\Delta\sigma)= \max_{\substack{
0 \leq F \leq \mathbbm{1} \\
\Tr(F)=k\\
F = \Delta F
}}   \Tr(F\sigma).
\end{equation}
Using these expressions above, one can rewrite 
\begin{equation}\label{Eq::Difference k-largest eigenvalues in passive states and coarse-grained states}
   S_{k-1}(\cg{\Lambda}(\rho))-S_{k-1}(\Delta\rho) =  \max_{\substack{
0 \leq Q \leq \mathbbm{1} \\
\Tr(Q)=k
}}   \Tr(Q\cg{\Lambda}(\rho)) -  \max_{\substack{
0 \leq F \leq \mathbbm{1} \\
\Tr(F)=k\\
F = \Delta F
}}   \Tr(F\rho).
\end{equation}
From the self-adjointness of the map $\cg{\Lambda}$, we get $\Tr(Q\cg{\Lambda}(\rho)) = \Tr(\cg{\Lambda}(Q)\rho)$. The CPTP and unital condition of $\cg{\Lambda}$ yields $\Tr(\cg{\Lambda}(Q)) = \Tr(Q) = k$, and $0\leq\cg{\Lambda}(Q)\leq \Id$ for $0\leq Q\leq \Id$. Similarly, one can show that $\Tr(\Delta(\cg{\Lambda}(Q))) = k$ and $0\leq\Delta( \cg{\Lambda}(Q))\leq \Id$. Hence, $\Delta( \cg{\Lambda}(Q))$ satisfies the feasibility condition  (see Eq.~\eqref{Eq::k-largest eigenvalue diagonal}) of the maximising operator $F$ for any $Q: 0\leq Q \leq \Id$ and $\Tr(Q)=k$. Substituting $\Delta( \cg{\Lambda}(Q))$ in Eq.~\eqref{Eq::Difference k-largest eigenvalues in passive states and coarse-grained states}, 
\begin{align}\label{Eq::no_name_yet2}
    &S_{k-1}(\cg{\rho})-S_{k-1}(\Delta\rho)~\nonumber\\
    &\leq  \max_{\substack{
0 \leq Q \leq \mathbbm{1} \\
\Tr(Q)=k
}}   \left[\Tr(\cg{\Lambda}(Q)\rho) - \Tr(\Delta( \cg{\Lambda}(Q))\rho)\right]\nonumber\\
&= \max_{\substack{
0 \leq Q \leq \mathbbm{1} \\
\Tr(Q)=k
}}   \left[\Tr(\cg{\Lambda}(Q) - \Delta( \cg{\Lambda}(Q)))\rho\right].
\end{align}
Now, maximising both sides over all states $\rho$, one gets
\begin{align}\label{Eq::no_name_yet3}
 &\max_\rho\left[S_{k-1}(\cg{\Lambda}(\rho))-S_{k-1}(\Delta\rho)\right]~\nonumber\\  
 &\leq \max_{\substack{
0 \leq Q \leq \mathbbm{1} \\
\Tr(Q)=k
}} \max_\rho  \left[\Tr(\cg{\Lambda}(Q) - \Delta( \cg{\Lambda}(Q)))\rho\right]\nonumber\\
&= \max_{\substack{
0 \leq Q \leq \mathbbm{1} \\
\Tr(Q)=k
}} \lambda_{\text{max}}  \left[\cg{\Lambda}(Q) - \Delta( \cg{\Lambda}(Q))\right].
\end{align}
The equality in the last step comes from the fact that maximising the expectation value of an operator over all states reproduces its maximum eigenvalue, i.e., $\max_\rho \Tr(\rho A) = \lambda_{\text{max}}(A)$. Now combining both Eq.~\eqref{Eq::no_name_yet1} and \eqref{Eq::no_name_yet3}, we complete the proof 
\begin{align}\label{Eq::no_name_yet4}
    A_M^H \leq \sum_{k=1}^{d-1} (E_k-E_{k-1})\max_{\substack{
0 \leq Q \leq \mathbbm{1} \\
\Tr(Q)=k
}} \lambda_{\text{max}}  \left[\cg{\Lambda}(Q) - \Delta( \cg{\Lambda}(Q))\right]
\end{align}
\end{proof}


\begin{lemma}\label{Lemma::Eigenvalue upper bound}
    For any effect $0 \le A\le\mathbbm{1}$, the maximum eigenvalue $\lambda_{\rm max}(A-\Delta A)$ is upper bounded by $1-{1}/{d}$. \label{Le:A_del_A_eigen_max_b}
\end{lemma}

For the proof of Lemma \ref{Le:A_del_A_eigen_max_b}, we refer the readers to Appendix \ref{App:proof_lemma_1}.
Now, we are ready to state and prove the relation between the robustness of the measurement coherence, and \OEA.

\begin{theorem}\label{Th::Coherence upper bounded by robustness} 
    For any measurement $M$, the \OEA is upper bounded by its coherence robustness (defined in Eq.~\eqref{Eq::Robustness of coherence}) multiplied with a scaling factor that depends on the dimension, and the spectral width of the Hamiltonian. More precisely, for an arbitrary $M\in \mathscr{M}(\cH_d, n)$
    \begin{equation}
        A_M^H \leq \frac{d-1}{d} (E_{d-1}-E_0)R_C(M).\label{Eq:OEA_Robust_relation}
    \end{equation}
\end{theorem}
\begin{proof}
    From the definition of Robustness of coherence in Eq.~\eqref{Eq::Robustness of coherence}, we have $M(\mu)+rP(\mu) = (1+r)N(\mu)$ for any outcome $\mu$ and any feasible $r$. For any $Q: 0\leq Q \leq \Id$ and $\Tr(Q)=k$, we have 
    \begin{align}
     &\cg{\Lambda}(Q) - \Delta( \cg{\Lambda}(Q)) \nonumber\\
     &= \sum_\mu \frac{\Tr(M(\mu)Q)}{\Tr(M(\mu))} (M(\mu)-\Delta(M(\mu)))\nonumber\\
     &= -r\sum_\mu q_\mu (P(\mu)-\Delta(P(\mu)))\nonumber\\
     &=r\left(\Id-\sum_\mu q_\mu P(\mu)\right)-r\left(\Id-\sum_\mu q_\mu \Delta(P(\mu))\right)\nonumber\\
     &= r(K_P-\Delta(K_P)).\nonumber
    \end{align}
In the second step, we used $q_\mu = {\Tr(M(\mu)Q)}/{\Tr(M(\mu))}$. Now note that $0\leq {\Tr(M(\mu)Q)}/{\Tr(M(\mu))}\leq {\Tr(M(\mu)I)}/{\Tr(M(\mu))}$, or $0\leq q_\mu\leq 1$. In fourth step, we defined $K_P = \Id-\sum_\mu q_\mu P(\mu)$. Since $\sum_\mu P(\mu) =\Id$, and $0\leq q_\mu\leq1$, we have $0\leq K_P\leq \Id$ and using Lemma~\ref{Lemma::Eigenvalue upper bound}, we get
\begin{equation}\label{Eq::Eigenvalue maximum upper bound}
   \lambda_{\text{max}}  \left[\cg{\Lambda}(Q) - \Delta( \cg{\Lambda}(Q))\right] \leq r\left(1-\frac{1}{d}\right).
\end{equation}
Since Eq.~\eqref{Eq::Eigenvalue maximum upper bound} is valid for any feasible $Q$ and $r$, 
\begin{align}
&\max_{\substack{
0 \leq Q \leq \mathbbm{1} \\
\Tr(Q)=k
}} \lambda_{\text{max}}  \left[\cg{\Lambda}(Q) - \Delta\left( \cg{\Lambda}(Q))\right] \leq R_C(M)(1-\frac{1}{d}\right), \\
&\text{or equivalently,}~ A_M^H \leq R_C(M)\left(1-\frac{1}{d}\right)\sum_{k=1}^{d-1} (E_k-E_{k-1})\\
&\implies~~ A_M^H \leq R_C(M)\left(1-\frac{1}{d}\right)\left(E_{d-1}-E_0\right).
\end{align}
\end{proof}
In the following example, we show that the upper bound of \OEA given in Eq. \eqref{Eq:OEA_Robust_relation} is tight for qubits.

\begin{example}\label{Ex:Tightness for upper bound}[Tightness of Eq.~\eqref{Eq:OEA_Robust_relation}]
\rm{For qubit measurements, the upper bound of \OEA in terms of robustness is tight for every projective measurement. In particular,
\[
A_M^H=\frac{E}{2}R_C(M)~~~\forall M \in \mathscr{P}(\cH_2,2),
\]
for the qubit Hamiltonian $H=E\ketbra{1}{1}$.
As discussed in Remark~\ref{Ratio advantage does not work}, an arbitrary qubit state and a rank-$1$
projective measurement can be written as $\rho=\left(\mathbbm{1}+\vec n\cdot\vec{\sigma}\right)/2$, and
   $ M(0)=\left(\mathbbm{1}+\vec m\cdot\vec{\sigma}\right)/2$,
    $M(1)=\left(\mathbbm{1}-\vec m\cdot\vec{\sigma}\right)/2$, respectively, where $|\vec n|\leq 1$ and $|\vec m|=1$. For this state and measurement, we have
$\mathcal{E_O}(\rho,M)={E}\left(|\vec n\cdot\vec m|-n_z\right)/2$, whereas
$\mathcal{E}_{\rm inc}(\rho)={E}\left(|n_z|-n_z\right)/2$ (see Appendix~\ref{App:observational ergotropic advantage for qubits}). Therefore,
\begin{equation}
    A_M^H
    =
    \frac{E}{2}
    \max_{|\vec n|\leq 1}
    \left(
    |\vec n\cdot\vec m|-|n_z|
    \right).
\end{equation}
Using the triangle inequality followed by the Cauchy--Schwarz inequality, we have
\begin{align*}
    |\vec n\cdot\vec m|
    &= |n_xm_x+n_ym_y+n_zm_z| \\
    &\leq
    |n_xm_x+n_ym_y|+|n_zm_z| \\
    &\leq
    \sqrt{n_x^2+n_y^2}\sqrt{m_x^2+m_y^2}
    +|n_z||m_z|.
\end{align*}
Therefore,
\begin{align*}
    |\vec n\cdot\vec m|-|n_z|
    &\leq
    \sqrt{n_x^2+n_y^2}\sqrt{m_x^2+m_y^2}
    -
    |n_z|(1-|m_z|) \\
    &\leq
    \sqrt{m_x^2+m_y^2}.
\end{align*}
In the second inequality, we have used $\sqrt{n_x^2+n_y^2}\leq 1$, which follows from $|\vec n|\leq1$, together with $1-|m_z|\geq0$, since $|\vec m|=1$. Thus, the second term is non-positive and can be dropped when obtaining an upper bound. The bound is achieved by choosing $n_z=0$ and $(n_x,n_y)$ parallel to $(m_x,m_y)$. Hence,
\begin{equation*}
    A_M^H
    =
    \frac{E}{2}\sqrt{m_x^2+m_y^2}.
\end{equation*}
For qubit measurements, the robustness of measurement coherence is equal to the $l_\infty$-norm based measure, i.e., $R_C(M)=C_{l_\infty}(M)$~\cite{Baek_2020_Quantifying}. Since there is only one pair of energy
eigenstates in the qubit case, from Eq.~\eqref{Eq::Infinity-norm based coherence} we have
\begin{align*}
    C_{l_\infty}(M)
    &=
    \sum_{\mu=0}^{1}
    \left|
    \bra{0}M(\mu)\ket{1}
    \right| \\
    &=
    \frac{1}{2}\sqrt{m_x^2+m_y^2}
    +
    \frac{1}{2}\sqrt{m_x^2+m_y^2} \\
    &=
    \sqrt{m_x^2+m_y^2}.
\end{align*}
Therefore,
\begin{equation*}
    A_M^H
    =
    \frac{E}{2}C_{l_\infty}(M)
    =
    \frac{E}{2}R_C(M).
\end{equation*}
Hence, every rank-$1$ projective qubit measurement saturates the upper bound in Theorem~3.}
\end{example}

\begin{theorem}\label{Th::OEA lower bound}
For any measurement $M$, the \OEA is lower bounded by its coherence, quantified by the $C_{l_\infty}$ norm measure of coherence, multiplied with the minimum energy-gap of the Hamiltonian with a dimensional factor. Precisely for $M\in \mathscr{M}(\cH_d, n)$, 
\begin{equation}\label{Eq::OEA lower bound}
   A_M^H \geq \frac{\min_k(E_k-E_{k-1})}{d^2}[C_{l_\infty}(M)]^2 
\end{equation}
\end{theorem}
\begin{proof}
From the definition of $C_{l_\infty}(M)$ in Eq.~\eqref{Eq::Infinity-norm based coherence}, let the maximum of the absolute value is obtained for a pair of energy eigenstates $\ket{E_i}$ and $\ket{E_j}$ with $i<j$ such that 
\begin{equation}
    C_{l_\infty}(M)
    =
    \sum_\mu
    \left|
    \bra{E_i}M(\mu)\ket{E_j}
    \right|
\end{equation}
Now for any phase $\phi$, we define
\begin{equation*}
    X_\phi
    =
    e^{i\phi}\ketbra{E_i}{E_j}
    +
    e^{-i\phi}\ketbra{E_j}{E_i}.
\end{equation*}
Since $X_\phi$ has eigenvalues $\{1,-1,0,\ldots,0\}$, the states
\begin{equation}
    \rho_{\phi,\pm}
    =
    \frac{1}{d}
    \left(
    \mathbbm{1}\pm X_\phi
    \right)
\end{equation}
are valid density matrices. Moreover,
\begin{equation}
    \Delta\rho_{\phi,\pm}
    =
    \frac{\mathbbm{1}}{d}.
\end{equation}

Using the definition of the coarse-graining map in Eq.~\eqref{Eq::Coarse-graining channel},
\begin{align}
    \Tr\left[
    X_\phi\Lambda_M^{\rm cg}(X_\phi)
    \right]
    &=
    \sum_\mu
    \frac{
    \left[
    \Tr(M(\mu)X_\phi)
    \right]^2
    }{\Tr(M(\mu))}.
\end{align}
Averaging the above expression over $\phi$ gives
\begin{equation}
    \frac{1}{2\pi}
    \int_0^{2\pi}
    d\phi\,
    \Tr\left[
    X_\phi\Lambda_M^{\rm cg}(X_\phi)
    \right]
    =
    2
    \sum_\mu
    \frac{
    |\langle E_i|M(\mu)|E_j\rangle|^2
    }{\Tr(M(\mu))}.
\end{equation}
Using the Cauchy--Schwarz inequality together with
$\sum_\mu\Tr(M(\mu))=d$, we have
\begin{align}
    \left[C_{l_\infty}(M)\right]^2
    &=
    \left(
    \sum_\mu
    |\langle E_i|M(\mu)|E_j\rangle|
    \right)^2
    \nonumber\\
    &\leq
    d
    \sum_\mu
    \frac{
    |\langle E_i|M(\mu)|E_j\rangle|^2
    }{\Tr(M(\mu))}.
\end{align}
Therefore, there exists a phase $\phi$ such that
\begin{equation}
    \Tr\left[
    X_\phi\Lambda_M^{\rm cg}(X_\phi)
    \right]
    \geq
    \frac{2}{d}
    \left[C_{l_\infty}(M)\right]^2.
    \label{eq:lower_bound_intermediate}
\end{equation}

Since $X_\phi$ has one eigenvalue $+1$ and one eigenvalue $-1$,
Eq.~\eqref{eq:lower_bound_intermediate} implies that, for one of the
two choices of sign,
\begin{equation}
    \lambda_{\max}
    \left(
    \pm\Lambda_M^{\rm cg}(X_\phi)
    \right)
    \geq
    \frac{1}{d}
    \left[C_{l_\infty}(M)\right]^2.
\end{equation}
We choose the corresponding state $\rho_{\phi,\pm}$. Since
$\Lambda_M^{\rm cg}$ is unital,
\begin{equation}
    \rho_M^{\rm cg}
    =
    \frac{\mathbbm{1}}{d}
    \pm
    \frac{1}{d}
    \Lambda_M^{\rm cg}(X_\phi),
\end{equation}
and consequently
\begin{equation}
    \lambda_{\max}(\rho_M^{\rm cg})
    -
    \frac{1}{d}
    \geq
    \frac{1}{d^2}
    \left[C_{l_\infty}(M)\right]^2.
    \label{eq:lambda_lower_bound}
\end{equation}

Using the expression for the passive energy introduced in the proof of
Theorem~\ref{Th::Coherence upper bound intermediate},
\begin{equation}
    P_H(\sigma)
    =
    E_{d-1}
    -
    \sum_{k=1}^{d-1}
    (E_k-E_{k-1})S^{k-1}(\sigma),
\end{equation}
we obtain
\begin{align}
    P_H\left(\frac{\mathbbm{1}}{d}\right)-P_H(\sigma)
    &=
    \sum_{k=1}^{d-1}
    (E_k-E_{k-1})
    \left[
    S^{k-1}(\sigma)-\frac{k}{d}
    \right]
    \nonumber\\
    &\geq
    \min_k(E_k-E_{k-1})
    \left[
    \lambda_{\max}(\sigma)-\frac{1}{d}
    \right],
\end{align}
where we used the fact that every state majorises the maximally mixed
state, so that
$S^{k-1}(\sigma)\geq k/d$ for every $k$.

For the state chosen above,
$\Delta\rho_{\phi,\pm}=\mathbbm{1}/d$. Hence, from
Eqs.~\eqref{Eq::Observational ergotropic advantage passive} and~\eqref{eq:lambda_lower_bound},
\begin{align}
    A_M^H
    &\geq
    P_H(\Delta\rho_{\phi,\pm})
    -
    P_H(\rho_M^{\rm cg})
    \nonumber\\
    &\geq
    \frac{\min_k(E_k-E_{k-1})}{d^2}
    \left[C_{l_\infty}(M)\right]^2.
\end{align}
This proves the desired lower bound.
\end{proof} 
\begin{example}\label{Ex:Tightness of lower bound}[Tightness of Eq.~\eqref{Eq::OEA lower bound}]
\begingroup
\setlength{\parindent}{0pt}
\setlength{\parskip}{0.5\baselineskip}
\rm The lower bound in Theorem~\ref{Th::OEA lower bound} is tight for qubit POVMs. To see this,
consider the Hamiltonian $H=E\ketbra{1}{1}$ and the three-outcome
measurement $M=\{M(0),M(1),M(2)\}$ with 
\begin{equation*}
    M(\mu)
    =
    \frac{1}{3}
    \left(
    \mathbbm{1}
    +
    c\,\Vec{m}_{\mu}\cdot\Vec{\sigma}
    \right),
    \qquad 0\leq c\leq 1,
\end{equation*}
where
\begin{equation*}
    \Vec{m}_0=(1,0,0),~~
    \Vec{m}_1=
    \left(-\frac{1}{2},\frac{\sqrt{3}}{2},0\right),~~
    \Vec{m}_2=
    \left(-\frac{1}{2},-\frac{\sqrt{3}}{2},0\right).
\end{equation*}
Since $\sum_{\mu}\Vec{m}_{\mu}=0$, the effects satisfy
$\sum_{\mu}M(\mu)=\mathbbm{1}$, and hence define a valid POVM.
For $c=1$, this reduces to the rank-$1$ trine measurement.

For each outcome,
\begin{equation*}
    \left|\bra{0}M(\mu)\ket{1}\right|
    =
    \frac{c}{3},
\end{equation*}
and therefore
\begin{equation*}
    C_{l_\infty}(M)
    =
    \sum_{\mu=0}^{2}
    \left|\bra{0}M(\mu)\ket{1}\right|
    =
    c.
\end{equation*}

Now consider an arbitrary qubit state
$\rho=\frac{1}{2}(\mathbbm{1}+\Vec{n}\cdot\Vec{\sigma})$.
The probability of outcome $\mu$ is
\begin{equation*}
    p(\mu)
    =
    \Tr[\rho M(\mu)]
    =
    \frac{1}{3}
    \left(
    1+c\,\Vec{n}\cdot\Vec{m}_{\mu}
    \right).
\end{equation*}
Since $\Tr[M(\mu)]=2/3$, the normalized effects are
\begin{equation*}
    \frac{M(\mu)}{\Tr[M(\mu)]}
    =
    \frac{1}{2}
    \left(
    \mathbbm{1}
    +
    c\,\Vec{m}_{\mu}\cdot\Vec{\sigma}
    \right).
\end{equation*}
Hence, the coarse-grained state is
\begin{align*}
    \rho_M^{\rm cg}
    &=
    \sum_{\mu=0}^{2}
    p(\mu)
    \frac{M(\mu)}{\Tr[M(\mu)]} \\
    &=
    \frac{1}{2}
    \left[
    \mathbbm{1}
    +
    \frac{c^2}{2}
    \left(
    n_x\sigma_x+n_y\sigma_y
    \right)
    \right],
\end{align*}
where we have used
\begin{equation*}
    \sum_{\mu=0}^{2}
    \Vec{m}_{\mu}\Vec{m}_{\mu}^{\,T}
    =
    \frac{3}{2}
    \begin{pmatrix}
        1 & 0 & 0\\
        0 & 1 & 0\\
        0 & 0 & 0
    \end{pmatrix}.
\end{equation*}
Therefore,
\begin{equation*}
    P_H(\rho_M^{\rm cg})
    =
    \frac{E}{2}
    \left(
    1-
    \frac{c^2}{2}\sqrt{n_x^2+n_y^2}
    \right).
\end{equation*}
On the other hand,
\begin{equation*}
    P_H(\Delta\rho)
    =
    \frac{E}{2}(1-|n_z|).
\end{equation*}
It follows that
\begin{equation*}
    P_H(\Delta\rho)-P_H(\rho_M^{\rm cg})
    =
    \frac{E}{2}
    \left[
    \frac{c^2}{2}\sqrt{n_x^2+n_y^2}
    -
    |n_z|
    \right].
\end{equation*}
The maximum is obtained for an equatorial pure state, for which
$n_z=0$ and $n_x^2+n_y^2=1$. Consequently,
\begin{equation*}
    A_M^H
    =
    \frac{Ec^2}{4}.
\end{equation*}
Since $C_{l_\infty}(M)=c$, we finally obtain
\begin{equation*}
    A_M^H
    =
    \frac{E}{4}
    \left[C_{l_\infty}(M)\right]^2,
\end{equation*}
which saturates the lower bound in Theorem~\ref{Th::OEA lower bound}.
\endgroup
\end{example}
An important point to emphasise is that, although we have shown that the observational ergotropic advantage is both upper and lower bounded by measures of measurement coherence, operationally it is a very different quantity. First, we make the following remark using a couple of examples.
\begin{remark}
\rm{The observational ergotropic advantage is not monotonic with the
robustness of measurement coherence, even for qubit measurements.
In particular, measurements with larger robustness can have smaller
observational ergotropic advantage, and measurements with the same
robustness can have different observational ergotropic advantages.

As a first counterexample, consider the qubit Hamiltonian
$H=E\ketbra{1}{1}$ and the two measurements
\begin{align*}
    M_1(0)
    &=
    \frac{1}{2}
    \left(
    \mathbbm{1}
    +
    \frac{3}{5}\sigma_x
    +
    \frac{4}{5}\sigma_z
    \right),
    \qquad
    M_1(1)=\mathbbm{1}-M_1(0), \\
    M_2(0)
    &=
    \frac{1}{2}
    \left(
    \mathbbm{1}
    +
    \frac{7}{10}\sigma_x
    \right),
    \qquad
    M_2(1)=\mathbbm{1}-M_2(0).
\end{align*}
Since for qubit measurements $R_C(M)=C_{l_\infty}(M)$~\cite{Baek_2020_Quantifying}, we immediately
obtain
\begin{equation*}
    R_C(M_1)=\frac{3}{5},
    \qquad
    R_C(M_2)=\frac{7}{10}.
\end{equation*}
The measurement $M_1$ is rank-$1$ projective and hence, from
Example~\ref{Ex:Tightness for upper bound},
\begin{equation*}
    A_{M_1}^H=\frac{3E}{10}.
\end{equation*}
For $M_2$, an arbitrary input state
$\rho=(\mathbbm{1}+\vec n\cdot\vec{\sigma})/2$ gives
\begin{equation*}
    \rho_{M_2}^{\rm cg}
    =
    \frac{1}{2}
    \left(
    \mathbbm{1}
    +
    \frac{49}{100}n_x\sigma_x
    \right).
\end{equation*}
Therefore, a direct maximisation of
$P_H(\Delta\rho)-P_H(\rho_{M_2}^{\rm cg})$ gives
\begin{equation*}
    A_{M_2}^H=\frac{49E}{200}.
\end{equation*}
Consequently,
\begin{equation*}
    R_C(M_1)<R_C(M_2),
    ~~~ \text{while}~~~~
    A_{M_1}^H>A_{M_2}^H.
\end{equation*}

As a second counterexample, consider $M_1$ above together with
\begin{equation*}
    M_3(0)
    =
    \frac{1}{2}
    \left(
    \mathbbm{1}
    +
    \frac{3}{5}\sigma_x
    \right),
    \qquad
    M_3(1)=\mathbbm{1}-M_3(0).
\end{equation*}
Again using $R_C(M)=C_{l_\infty}(M)$ for qubits, we have
\begin{equation*}
    R_C(M_1)=R_C(M_3)=\frac{3}{5}.
\end{equation*}
For $M_3$, the corresponding coarse-grained state is
\begin{equation*}
    \rho_{M_3}^{\rm cg}
    =
    \frac{1}{2}
    \left(
    \mathbbm{1}
    +
    \frac{9}{25}n_x\sigma_x
    \right),
\end{equation*}
which gives
\begin{equation*}
    A_{M_3}^H=\frac{9E}{50}.
\end{equation*}
Hence,
\begin{equation*}
    R_C(M_1)=R_C(M_3),
    ~~~\text{while}~~~~
    A_{M_1}^H=\frac{3E}{10}
    \neq
    \frac{9E}{50}=A_{M_3}^H.
\end{equation*}
These counterexamples show that the robustness of measurement
coherence does not determine the ordering, or even the value, of the
observational ergotropic advantage.}
\end{remark}
In addition to that, we also note that \OEA itself does not constitute a valid coherence monotone. The relevant measures of measurement coherence are monotonically non-increasing under free operations in the corresponding resource theory, such as strictly incoherent operations (SIO). In contrast, we find that the observational ergotropic advantage does not satisfy such a monotonicity property. In particular, in Appendix~\ref{Counter example of the PIO}, we provide a counterexample showing that this quantity fails to be monotonic even under physically incoherent operations (PIO), which constitute one of the most restrictive classes of free operations~\cite{chitambar2016critical}. Therefore, despite its close connection to established measures of measurement coherence through the derived bounds, the observational ergotropic advantage cannot itself serve as a coherence monotone for quantifying measurement coherence.

\section{Conclusion}
\label{Sec:conc}

In this work, we have studied the quantitative relationship between observational ergotropic advantage and measurement coherence by introducing a quantifier of the former and deriving upper and lower bounds of it in terms of robustness of measurement coherence and the $l_{\infty}$-norm-based measure. We have observed that measurement coherence is not only necessary, but also sufficient for observational ergotropic advantage. Furthermore, we have proved the tightness of these bounds for qubit cases. Interestingly, we  have found that more measurement coherence does not imply more observational ergotropic advantage, in general and also have shown that the advantage quantifier is not monotone under the resource theory of measurement coherence.  This study will serve as the foundational understanding on the measurement coherence as a resource in extracting work from unknown quantum systems.

Our work opens up several interesting research directions. The first immediate question is: under which resource theory, our observational ergotropic advantage quantifier is a monotone? Secondly, in this work, as we have restricted ourselves to only the thermodynamic application of measurement coherence, it is also important to explore the application of measurement coherence in some information-theoretic tasks or some quantum algorithms. Thirdly, this work can easily be generalized for any degenerate Hamiltonian which involves the block-measurement coherence.

\section{Acknowledgement}
SM acknowledges funding from Sydney Quantum Academy. AM acknowledges funding from STARS (Grant No. STARS/STARS-2/2023-0809), Government of India.

\bibliography{references}

\appendix
\section{Proof of Lemma 1}
\label{App:proof_lemma_1}
\begin{proof}
Let $A$ be an arbitrary effect, i.e., $0\leq A\leq\mathbbm{1}$, and
let $\{\ket{E_i}\}_{i=0}^{d-1}$ denote the energy eigenbasis in which
the dephasing map $\Delta$ is defined. For an arbitrary unit vector
$\ket{v}=\sum_i c_i\ket{E_i}$, we have
\begin{align*}
    \bra{v}A\ket{v}
    &=
    \left\|A^{1/2}\ket{v}\right\|^2 \\
    &=
    \left\|
    \sum_i c_i A^{1/2}\ket{E_i}
    \right\|^2.
\end{align*}
Using the Cauchy--Schwarz inequality,
\begin{equation*}
    \left\|\sum_i \ket{x_i}\right\|^2
    \leq
    d\sum_i \|\ket{x_i}\|^2,
\end{equation*}
for any collection of $d$ vectors $\{\ket{x_i}\}$, we obtain
\begin{align*}
    \bra{v}A\ket{v}
    &\leq
    d\sum_i |c_i|^2
    \left\|A^{1/2}\ket{E_i}\right\|^2 \\
    &=
    d\sum_i |c_i|^2
    \bra{E_i}A\ket{E_i}.
\end{align*}
Since
\begin{equation*}
    \Delta(A)
    =
    \sum_i
    \bra{E_i}A\ket{E_i}
    \ketbra{E_i}{E_i},
\end{equation*}
the above inequality can be written as
\begin{equation*}
    \bra{v}A\ket{v}
    \leq
    d\,\bra{v}\Delta(A)\ket{v}.
\end{equation*}
Equivalently,
\begin{equation*}
    \bra{v}\Delta(A)\ket{v}
    \geq
    \frac{1}{d}\bra{v}A\ket{v}.
\end{equation*}
It therefore follows that
\begin{align*}
    \bra{v}\bigl[A-\Delta(A)\bigr]\ket{v}
    &\leq
    \left(1-\frac{1}{d}\right)
    \bra{v}A\ket{v} \\
    &\leq
    1-\frac{1}{d},
\end{align*}
where the second inequality follows from $A\leq\mathbbm{1}$ and
$\bra{v}v\rangle=1$. Since the above relation holds for every unit
vector $\ket{v}$, maximising over all such vectors gives
\begin{equation*}
    \lambda_{\max}\left[A-\Delta(A)\right]
    \leq
    1-\frac{1}{d}.
\end{equation*}
This proves the result.
\end{proof}
Furthermore, we show that the bound is tight with the following example. Consider a maximally coherent
rank-$1$ projector
\begin{equation*}
    A=\ket{\psi}\bra{\psi},
    \qquad
    \ket{\psi}
    =
    \frac{1}{\sqrt{d}}
    \sum_{i=0}^{d-1} e^{\mathrm{i}\phi_i}\ket{E_i}.
\end{equation*}
For this effect,
\begin{equation*}
    \Delta(A)=\frac{\mathbbm{1}}{d},
\end{equation*}
and hence
\begin{equation*}
    \lambda_{\max}\left[A-\Delta(A)\right]
    =
    \lambda_{\max}\left[
    \ketbra{\psi}{\psi}-\frac{\mathbbm{1}}{d}
    \right]
    =
    1-\frac{1}{d}.
\end{equation*}
Note that, this also shows that the dimensional prefactor appearing in Theorem~\ref{Th::Coherence upper bounded by robustness} is optimal.
\section{Observational Ergotropy and Incoherent Ergotropy for a projective measurement on qubit states}
\label{App:observational ergotropic advantage for qubits}
\begingroup
\setlength{\parindent}{0pt}
\setlength{\parskip}{0.5\baselineskip}

Here, we explicitly calculate the observational ergotropy and the
incoherent ergotropy for a qubit state under a rank-1 projective
measurement. We consider the Hamiltonian
\begin{equation*}
    H=E\ketbra{1}{1}.
\end{equation*}
An arbitrary qubit state can be written in the Bloch representation as
\begin{equation*}
    \rho
    =
    \frac{1}{2}
    \left(
    \mathbbm{1}+\Vec{n}\cdot\Vec{\sigma}
    \right),
    \qquad
    |\Vec{n}|\leq 1,
\end{equation*}
where $\Vec{\sigma}=(\sigma_x,\sigma_y,\sigma_z)$. Similarly, an
arbitrary rank-$1$ projective measurement can be written as
\begin{equation*}
    M(0)
    =
    \frac{1}{2}
    \left(
    \mathbbm{1}+\Vec{m}\cdot\Vec{\sigma}
    \right),
    \qquad
    M(1)
    =
    \frac{1}{2}
    \left(
    \mathbbm{1}-\Vec{m}\cdot\Vec{\sigma}
    \right),
\end{equation*}
where $|\Vec{m}|=1$.

The probabilities corresponding to the two measurement outcomes are
\begin{align*}
    p(0)=\Tr[\rho M(0)]
    &=\frac{1}{2}\left(1+\Vec{n}\cdot\Vec{m}\right),
    \\
    p(1)=\Tr[\rho M(1)]
    &=\frac{1}{2}\left(1-\Vec{n}\cdot\Vec{m}\right).
\end{align*}
Since $\Tr[M(0)]=\Tr[M(1)]=1$, the corresponding coarse-grained state
is
\begin{align*}
    \rho_M^{\rm cg}
    &=
    p(0)M(0)+p(1)M(1)\\
    &=
    \frac{1}{2}
    \left[
    \mathbbm{1}
    +
    (\Vec{n}\cdot\Vec{m})\Vec{m}\cdot\Vec{\sigma}
    \right].
\end{align*}
Its eigenvalues are therefore
\begin{equation*}
    \lambda_{\pm}\left(\rho_M^{\rm cg}\right)
    =
    \frac{1}{2}
    \left(
    1\pm|\Vec{n}\cdot\Vec{m}|
    \right).
\end{equation*}
The passive state assigns the larger eigenvalue to the ground state and
the smaller eigenvalue to the excited state. Hence,
\begin{equation*}
    P_H\left(\rho_M^{\rm cg}\right)
    =
    \frac{E}{2}
    \left(
    1-|\Vec{n}\cdot\Vec{m}|
    \right).
\end{equation*}
On the other hand,
\begin{equation*}
    \Tr(H\rho)
    =
    E\bra{1}\rho\ket{1}
    =
    \frac{E}{2}(1-n_z).
\end{equation*}
Using the definition of observational ergotropy, we obtain
\begin{align*}
    \mathcal{E}_O(\rho,M)
    &=
    \Tr(H\rho)
    -
    P_H\left(\rho_M^{\rm cg}\right)\\
    &=
    \frac{E}{2}
    \left(
    |\Vec{n}\cdot\Vec{m}|-n_z
    \right).
\end{align*}

We next calculate the incoherent ergotropy. Dephasing $\rho$ in the
energy eigenbasis gives
\begin{equation*}
    \Delta\rho
    =
    \frac{1}{2}
    \left(
    \mathbbm{1}+n_z\sigma_z
    \right),
\end{equation*}
whose eigenvalues are
\begin{equation*}
    \lambda_{\pm}(\Delta\rho)
    =
    \frac{1}{2}
    \left(
    1\pm|n_z|
    \right).
\end{equation*}
Therefore, its passive energy is
\begin{equation*}
    P_H(\Delta\rho)
    =
    \frac{E}{2}
    \left(
    1-|n_z|
    \right).
\end{equation*}
Since $\Tr(H\Delta\rho)=\Tr(H\rho)=E(1-n_z)/2$, the incoherent
ergotropy becomes
\begin{align*}
    \mathcal{E}_{\rm inc}(\rho)
    &=
    \Tr(H\Delta\rho)-P_H(\Delta\rho)\\
    &=
    \frac{E}{2}
    \left(
    |n_z|-n_z
    \right).
\end{align*}
\endgroup
\section{Observational ergotropic advantage is not a monotone under any resource theory of measurement coherence}\label{Counter example of the PIO}

In the resource theory of quantum coherence, coherence is defined with respect to a fixed reference basis. As in any resource-theoretic framework, two central ingredients are the specification of the free states and the free operations. The free states are the incoherent states, whose density matrices are diagonal in this basis, whereas states containing off-diagonal elements are regarded as resourceful. Similarly, free operations are those that cannot generate coherence from incoherent states, with different choices of allowed operations. Few examples are maximally incoherent operations (MIO), incoherent operations (IO), strictly incoherent operations (SIO), and physically incoherent operations (PIO) which gives rise to different resource-theoretic frameworks. A valid coherence measure is therefore required to vanish on incoherent states and to be non-increasing under the chosen class of free operations.

This framework can be extended from quantum states to quantum measurements by treating a POVM $A=\{A(a)\}$ as the resource object. A measurement is said to be incoherent if its outcome statistics are insensitive to the coherence present in the input state, which is equivalent to requiring each POVM element $A(a)$ to be diagonal in the reference basis. Coherent measurements are therefore characterized by non-vanishing off-diagonal components and can access the coherence of quantum states. In the resource theory of measurement coherence introduced in Ref.~\cite{Baek_2020_Quantifying}, the free transformations are induced by the dual action of strictly incoherent operations (SIO), and a legitimate measure of measurement coherence must be non-increasing under such transformations.

Among the different classes of incoherent operations, physically incoherent operations (PIO) constitute a particularly restrictive class and were introduced to characterize operations that admit a fully resource-free physical implementation~\cite{chitambar2016critical}. Their Kraus operators can be written in the form $K_j = U_j P_j$, where $\{P_j\}$ forms a set of mutually orthogonal incoherent projectors and $U_j$ are incoherent unitaries. In the measurement setting, a PIO acts through the corresponding dual map,
\begin{equation}
    M(\mu) \longmapsto \sum_j K_j^\dagger M(\mu) K_j,
\end{equation}
and hence provides a particularly stringent class of free transformations for testing whether a proposed quantifier of measurement coherence satisfies the required monotonicity condition.

We now show that the observational ergotropic advantage is not monotonic under arbitrary physically incoherent operations (PIO)~\cite{chitambar2016critical}. In particular, we construct a PIO $\Lambda$ and a qutrit measurement $M$ such that
\begin{equation}
A_{\Lambda^\dagger(M)}^H>A_M^H.
\end{equation}

Consider a three-dimensional system with Hamiltonian
\begin{equation}\label{hamiltonian counter example}
H = 0\ketbra{E_0}{E_0} + \ketbra{E_1}{E_1} + 2\ketbra{E_2}{E_2},
\end{equation}
where the energy eigenbasis is taken as the incoherent basis.

Define the channel
\begin{equation}
\Lambda(\rho) = K_1\rho K_1^\dagger + K_2\rho K_2^\dagger,
\end{equation}
with Kraus operators
\begin{equation}
K_1 =
\begin{pmatrix}
1&0&0\\
0&1&0\\
0&0&0
\end{pmatrix},
\qquad
K_2 =
\begin{pmatrix}
0&0&1\\
0&0&0\\
0&0&0
\end{pmatrix}.
\end{equation}
Clearly,
\begin{equation}
K_1^\dagger K_1+K_2^\dagger K_2=\mathbbm{1},
\end{equation}
so $\Lambda$ is trace preserving.

Moreover, $\Lambda$ is a PIO. To see this, define the incoherent projectors
\begin{equation}
P_1 = \ketbra{E_0}{E_0} + \ketbra{E_1}{E_1}, \qquad P_2 = \ketbra{E_2}{E_2},
\end{equation}
and the incoherent unitaries
\begin{equation}
U_1=\mathbbm{1}, \qquad U_2 = \ketbra{E_0}{E_2} + \ketbra{E_1}{E_1} + \ketbra{E_2}{E_0}.
\end{equation}
Then
\begin{equation}
K_1=U_1P_1, \qquad K_2=U_2P_2,
\end{equation}
which is precisely the Kraus form of a physically incoherent operation~\cite{chitambar2016critical}. 

Now consider the two-outcome measurement
\begin{equation}
M=\{M(0),M(1)\}, \qquad M(1)=\mathbbm{1}-M(0),
\end{equation}
where
\begin{equation}
M(0) =
\begin{pmatrix}
\frac{11}{50} & -\frac{7}{20} & 0\\
-\frac{7}{20} & \frac{29}{50} & 0\\
0&0&\frac{29}{50}
\end{pmatrix}.
\end{equation}
The eigenvalues of $M(0)$ and $M(1)$ lie in $[0,1]$. Therefore, $M$ is a valid POVM.

The dual action of $\Lambda$ is
\begin{equation}
\Lambda^\dagger(X) = K_1^\dagger XK_1 + K_2^\dagger XK_2.
\end{equation}
Consequently, for
\begin{equation}
M':=\Lambda^\dagger(M),
\end{equation}
we obtain
\begin{equation}
M'(0) =
\begin{pmatrix}
\frac{11}{50} & -\frac{7}{20} & 0\\
-\frac{7}{20} & \frac{29}{50} & 0\\
0&0&\frac{11}{50}
\end{pmatrix},
\qquad M'(1)=\mathbbm{1}-M'(0).
\end{equation}
Thus, $M'$ is also a valid POVM.

It remains to compare $A_M^H$ and $A_{M'}^H$. Since both quantities are defined through a maximization over all input states $\rho$, the main task is to determine their global maxima. Rather than performing a numerical optimization over the full set of qutrit density matrices, we carry out the maximization analytically by parameterizing the relevant degrees of freedom of $\rho$. This allows us to identify a maximizing state $\rho^*$ and to obtain exact expressions for both $A_M^H$ and $A_{M'}^H$. To this end, it is useful to first derive a general expression for an arbitrary binary qutrit measurement. Let

\begin{equation}
N=\{N(0),\mathbbm{1}-N(0)\},
\end{equation}
and define
\begin{equation}
V:=\Tr[N(0)], \qquad t:=\Tr[\rho N(0)].
\end{equation}
Since the system is three dimensional,
\begin{equation}
\Tr[\mathbbm{1}-N(0)]=3-V,
\end{equation}
and since the measurement has only two outcomes, the probability of the second outcome is $1-t$. Therefore, the coarse-grained state is
\begin{align}
\rho_N^{cg} &= t\frac{N(0)}{V} + (1-t)\frac{\mathbbm{1}-N(0)}{3-V} \nonumber\\
&= \frac{1-t}{3-V}\mathbbm{1} + \left( \frac{t}{V} - \frac{1-t}{3-V} \right)N(0).
\label{qutrit cg}
\end{align}

Equation~\eqref{qutrit cg} shows that $\rho_N^{cg}$ is an affine function of $N(0)$. Hence, $N(0)$ and $\rho_N^{cg}$ have the same eigenvectors. If $m_i$ are the eigenvalues of $N(0)$, the corresponding eigenvalues of $\rho_N^{cg}$ are
\begin{equation}
\lambda_i(\rho_N^{cg}) = \frac{1-t}{3-V} + \left( \frac{t}{V} - \frac{1-t}{3-V} \right)m_i.
\end{equation}
We now define the spectral width of $N(0)$ as
\begin{equation}
D_N := \lambda_{\max}[N(0)] - \lambda_{\min}[N(0)].
\end{equation}
It follows that the spectral width of the coarse-grained state is
\begin{equation}
\lambda_{\max}(\rho_N^{cg}) - \lambda_{\min}(\rho_N^{cg}) = D_N \left| \frac{t}{V} - \frac{1-t}{3-V} \right|.
\end{equation}
The absolute value is required because the coefficient multiplying $N(0)$ can be either positive or negative. Simplifying the coefficient gives
\begin{equation}
\lambda_{\max}(\rho_N^{cg}) - \lambda_{\min}(\rho_N^{cg}) = D_N \frac{|3t-V|}{V(3-V)}.
\label{eq:cg-spectral-width}
\end{equation}

For the considered Hamiltonian \eqref{hamiltonian counter example} here, the passive energy of an arbitrary qutrit state $\omega$ takes a particularly simple form. Let the eigenvalues of $\omega$ be ordered as
\begin{equation}
\lambda_{\max}(\omega) \geq \lambda_{\mathrm{mid}}(\omega) \geq \lambda_{\min}(\omega).
\end{equation}
The passive state assigns the largest eigenvalue to the lowest energy, the middle eigenvalue to the middle energy, and the smallest eigenvalue to the highest energy. Hence,
\begin{align}
P_H(\omega) &= \lambda_{\mathrm{mid}}(\omega) + 2\lambda_{\min}(\omega) \nonumber\\
&= 1- \left[ \lambda_{\max}(\omega) - \lambda_{\min}(\omega) \right].
\label{eq:qutrit-passive-width}
\end{align}
Thus, for the equally spaced qutrit Hamiltonian, the passive energy depends only on the spectral width of the state.

Using Eq.~\eqref{eq:cg-spectral-width}, we immediately obtain
\begin{equation}\label{eq:passive-cg-general}
P_H(\rho_N^{cg}) = 1- D_N\frac{|3t-V|}{V(3-V)}.
\end{equation}

Next, write the diagonal entries of the input state in the energy eigenbasis as
\begin{equation}
p_i:=\bra{E_i}\rho\ket{E_i}, \qquad p_0+p_1+p_2=1.
\end{equation}
The dephased state is
\begin{equation}
\Delta\rho = \sum_{i=0}^{2} p_i\ket{E_i}\bra{E_i},
\end{equation}
and its eigenvalues are simply $\{p_0,p_1,p_2\}$. Therefore,
\begin{equation}
P_H(\Delta\rho) = 1- (p_{\max}-p_{\min}),
\end{equation}
where
\begin{equation}
p_{\max}:=\max\{p_0,p_1,p_2\}, \qquad p_{\min}:=\min\{p_0,p_1,p_2\}.
\end{equation}

Combining this result with Eq.~\eqref{eq:passive-cg-general}, we find
\begin{equation}\label{reduced observational advantage}
P_H(\Delta\rho)-P_H(\rho_N^{cg}) = D_N \frac{|3t-V|}{V(3-V)} - (p_{\max}-p_{\min}).
\end{equation}
This expression is the key simplification: the first term depends on the measurement through $D_N$, $V$, and the outcome probability $t$, while the second term penalizes unequal populations in the input state.

We now specialize Eq.~\eqref{reduced observational advantage} to the measurements $M$ and $M'$. In both cases, the only nonzero off-diagonal entries occur in the $\{\ket{E_0},\ket{E_1}\}$ block. Therefore, the only coherence of $\rho$ that can contribute to $\Tr[\rho M(0)]$ or $\Tr[\rho M'(0)]$ is $\rho_{01}$. Since the off-diagonal entries of both effects are real, only the real part of $\rho_{01}$ is relevant. We define
\begin{equation}
x := \mathrm{Re} \bra{E_0}\rho\ket{E_1}.
\end{equation}

Positivity of $\rho$ implies positivity of every principal submatrix. In particular,
\begin{equation}
\begin{pmatrix}
p_0 & \rho_{01}\\
\rho_{01}^* & p_1
\end{pmatrix}
\geq0,
\end{equation}
which requires
\begin{equation}
p_0p_1-|\rho_{01}|^2\geq0.
\end{equation}
Hence,
\begin{equation}
|\rho_{01}|\leq\sqrt{p_0p_1},
\end{equation}
and therefore
\begin{equation}\label{eq:coherence-bound}
|x|\leq\sqrt{p_0p_1}.
\end{equation}

For the original measurement $M(0)$,
\begin{align}
t &= \Tr[\rho M(0)] \nonumber\\
&= \frac{11}{50}p_0 + \frac{29}{50}p_1 + \frac{29}{50}p_2 - \frac{7}{10}x.
\end{align}
Using $p_0+p_1+p_2=1$, this reduces to
\begin{equation}\label{eq:t-original-reduced}
t = \frac{29}{50} - \frac{9}{25}p_0 - \frac{7}{10}x.
\end{equation}
The corresponding volume is
\begin{equation}
V = \Tr[M(0)] = \frac{69}{50}.
\end{equation}

For the transformed measurement $M'(0)$,
\begin{align}
t' &= \Tr[\rho M'(0)] \nonumber\\
&= \frac{11}{50} + \frac{9}{25}p_1 - \frac{7}{10}x,
\label{eq:t-transformed-reduced}
\end{align}
where we used $p_0+p_2=1-p_1$. Its volume is
\begin{equation}
V' = \Tr[M'(0)] = \frac{51}{50}.
\end{equation}

The nontrivial $2\times2$ coherent block is identical for $M(0)$ and $M'(0)$. Its eigenvalues are
\begin{equation}
\frac{40\pm\sqrt{1549}}{100},
\end{equation}
and these respectively give the maximal and minimal eigenvalues of both effects. Consequently,
\begin{equation}
D_M = D_{M'} = \frac{\sqrt{1549}}{50}.
\label{eq:DMMprime}
\end{equation}

Substituting Eqs.~\eqref{eq:t-original-reduced}, \eqref{eq:t-transformed-reduced}, and \eqref{eq:DMMprime} into Eq.~\eqref{reduced observational advantage} reduces the optimization over arbitrary qutrit states to the four real parameters $p_0,p_1,p_2$, and $x$, subject to
\begin{equation}
p_i\geq0, \qquad p_0+p_1+p_2=1, \qquad |x|\leq\sqrt{p_0p_1}.
\end{equation}

For fixed populations, the dependence on $x$ appears only through the absolute value of an affine function of $x$. Hence, for fixed $p_0,p_1,p_2$, the maximum is attained at one of the boundary values allowed by Eq.~\eqref{eq:coherence-bound},
\begin{equation}
x=\pm\sqrt{p_0p_1}.
\end{equation}
Thus, the maximizing state uses the largest coherence between $\ket{E_0}$ and $\ket{E_1}$ compatible with its diagonal populations.

The second term in Eq.~\eqref{reduced observational advantage},
\begin{equation}
-(p_{\max}-p_{\min}),
\end{equation}
is always non-positive and vanishes only when all populations are equal. This suggests choosing
\begin{equation}
p_0=p_1=p_2=\frac{1}{3}.
\end{equation}
For these populations, the positivity bound becomes
\begin{equation}
|x|\leq\frac13,
\end{equation}
so that one can simultaneously eliminate the population penalty and retain the maximal allowed coherence. The global maximization of the reduced expression is attained for
\begin{equation}
p_0=p_1=p_2=\frac{1}{3}, \qquad x=-\frac{1}{3}.
\end{equation}
A density matrix realizing these parameters is
\begin{equation}
\rho_* =
\begin{pmatrix}
\frac{1}{3}&-\frac{1}{3}&0\\
-\frac{1}{3}&\frac{1}{3}&0\\
0&0&\frac{1}{3}
\end{pmatrix}.
\label{eq:optimal-state-PIO}
\end{equation}
The eigenvalues of $\rho_*$ are
\begin{equation}
\left\{ \frac{2}{3},\frac{1}{3},0 \right\},
\end{equation}
so $\rho_*$ is a valid density matrix. Moreover,
\begin{equation}
\Delta\rho_*=\frac{\mathbbm{1}}{3},
\end{equation}
and therefore
\begin{equation}
p_{\max}-p_{\min}=0.
\end{equation}

For this state, Eq.~\eqref{eq:t-original-reduced} gives
\begin{equation}
t=\frac{52}{75},
\end{equation}
and therefore
\begin{equation}
3t-V=\frac{7}{10}.
\end{equation}
Using Eq.~\eqref{reduced observational advantage},
\begin{align}
A_M^H &= D_M\frac{|3t-V|}{V(3-V)} \nonumber\\
&= \frac{35\sqrt{1549}}{5589} \nonumber\\
&\simeq 0.24647.
\end{align}

Similarly, for the transformed measurement, Eq.~\eqref{eq:t-transformed-reduced} gives
\begin{equation}
t'=\frac{43}{75},
\end{equation}
so that
\begin{equation}
3t'-V'=\frac{7}{10}.
\end{equation}
Hence,
\begin{align}
A_{M'}^H &= D_{M'}\frac{|3t'-V'|}{V'(3-V')} \nonumber\\
&= \frac{35\sqrt{1549}}{5049} \nonumber\\
&\simeq 0.27283.
\end{align}

Thus,
\begin{equation}
A_{M'}^H>A_M^H.
\end{equation}


\end{document}